\documentclass[a4paper,11pt]{amsart}
\usepackage{amssymb,amsmath,amsfonts,amsthm}
\usepackage[english]{babel}
\usepackage[utf8]{inputenc}
\usepackage{enumerate,expdlist}
\newcommand{\euler}{\mathrm{e}}
\newcommand{\RR}{\mathbb{R}}

\newcommand{\NN}{\mathbb{N}}	
\newcommand{\ZZ}{\mathbb{Z}}
\newcommand{\PP}{\mathbb{P}}
\newcommand{\EE}{\mathbb{E}}
\newtheorem{theorem}{Theorem}[section]
\newtheorem{lemma}[theorem]{Lemma}

\newtheorem{corollary}[theorem]{Corollary}
\theoremstyle{definition}
\newtheorem{definition}[theorem]{Definition}

\theoremstyle{remark}
\newtheorem{remark}[theorem]{Remark}
\usepackage{microtype}
\usepackage{ellipsis}
\usepackage{typearea}
\usepackage[textsize=tiny,shadow]{todonotes}
\begin{document}
\title{Wegner estimate for Landau-breather Hamiltonians}
\author{Matthias T\"aufer, Ivan Veseli\'c}
 \address{Technische Universit\"at Chemnitz, Fakult\"at f\"ur Mathematik, Germany}
%
\date{\today}

\begin{abstract}
 We consider Landau Hamiltonians with a weak coupling random electric potential of breather type.
 Under appropriate assumptions we prove a Wegner estimate.
 It implies the Hoelder continuity of the integrated density of states.
 The main challenge is the problem how to deal with non-linear dependence on the random parameters.
\end{abstract}

\maketitle

\section{Introduction}

Wegner estimates play an important role in the theory of random Schr\"odinger operators.
They provide upper bounds on the expected number of eigenvalues (or sometimes on the probability of finding at least one eigenvalue) in an energy interval of a finite volume restriction of a random operator, formalizing the intuitive fact that eigenvalues of a random operator are indeed random, and in many cases continuously distributed.
Such estimates are a crucial ingredient in the multi-scale analysis which is an inductive procedure over length scales to establish localization, i.e. the occurrence of dense pure point spectrum in some energy regimes for the full space operator and the non-spreading of initially localized wavepackets.
Furthermore, Wegner estimates allow to study the regularity of the integrated density of states (or spectral distribution function) associated to ergodic random operators, cf. Corollary~\ref{cor:IDS} for a precise statement.
\par
In this article,  we prove a Wegner estimate, Theorem~\ref{thm:Wegner}, for the Landau Hamiltonian with a random breather potential
 \begin{align} \label{e:HoB}
 H_B^\omega = H_B + \lambda V_\omega,\quad
 &\text{where}\
 H_B = (- i \nabla - A)^2,\
 A = (B/2) (x_2, -x_1)\
 \text{and}\\
 \nonumber
 & V_\omega(x) := \sum_{j \in \ZZ^2} u_{\omega_j}(x-j),\
 u_{t}(x) = u(x/t),\
 \omega \in \Omega \subset (0, 1/2)^{\ZZ^2}
 \end{align}
in the small disorder regime $0<\lambda \ll 1 $. Here $u: \RR^2 \to [0, \infty)$ is an appropriately chosen, bounded and compactly supported single-site potential, see hypotheses (i)--(iii) below and inequality \eqref{eq:partial_u}.
Corollary~\ref{cor:IDS} yields H\"older continuity of the integrated density of states with respect to any H\"older exponent $\theta \in (0,1)$ at low energies in the low disorder regime.
More precisely, for every $E_0 \in \RR$, there is a critical disorder parameter $\lambda_0 > 0$ such that for all $\lambda \in (0,\lambda_0)$, the IDS is H\"older continuous in $(- \infty, E_0]$.
\par
While Wegner estimates are well understood for many models, in particular the alloy-type model, the random breather Landau Hamiltonian poses additional challenges, namely unbounded coefficients functions of the partial differential operator and a non-linear dependence of the operator on the random parameters $\omega_j$.
The conjunction of these features does not allow to apply one of the established Wegner estimates to the model under consideration.
\par
Let us explain this in more detail:
If a family of electromagnetic random Schr\"odinger operators $\{H^\omega\}_{\omega \in \tilde \Omega}$, $\tilde \Omega$ being some probability space, has coefficients which are uniformly bounded in space and in randomness, then there are nowadays powerful \emph{scale free unique continuation principles} (UCPs for short) at disposal \cite{RojasMolinaV-13,Klein-13,NakicTTV-15,BorisovTV-15} which facilitate the proof of a Wegner estimate.
The term \emph{scale free} means that the UCP yields estimates for restrictions of $H^\omega$ to cubes of side length $L \in \NN$ which are independent of the scale $L$.
Furthermore, these UCPs are \emph{uniform} over the ensemble $\left\{ H^\omega\right\}_{ \omega \in \tilde \Omega}$.
While the first three cited papers allow only for electric potentials, \cite{BorisovTV-15} covers electro-magnetic vector potentials as well, yet still only for operators with uniformly bounded coefficient functions.
Our object of study, the operator $H_B^\omega$ in \eqref{e:HoB}, can be written as
\[
 H_B^\omega u = - \Delta u + 2 i A \cdot \nabla u + (A^2 + \lambda V_\omega) u,
\]
and we see that it has the zeroth-order coefficient function $(A^2 + \lambda V_\omega)$ and the first-order coefficient function $2 i A$. Neither of them is bounded on $\RR^2$.
Therefore, none of the above mentioned scale free UCPs apply.
Of course, once we restrict $H_B^\omega$ to a finite volume box the coefficients are indeed bounded, but with a sup norm which depends on the diameter of the box.
The resulting Wegner estimate would carry an exponential volume dependence which would be too weak to yield any useful application.
\par
On the other hand, if the random parameters influence the Hamiltonian in a \emph{linear} way, as it is the case for the alloy-type potential
\begin{equation*}
V_\omega(x) := \sum_{j \in \ZZ^2} \omega_j  \, u(x-j),
\end{equation*}
one can exploit the linearity and the eigenvalue equation \cite{CombesHK-03,CombesHK-07}
and find a Wegner estimate by merely using a UCP for the unperturbed, non-random part, cf.~Section 4 in \cite{CombesHK-03}.
\par
Since we are in none of the two described situations we have to devise new ideas to complete the Wegner estimate, in particular a novel deterministic estimate, Theorem \ref{prop:prop_1}, for traces of spectral projections valid for Hamiltonians satisfying an abstract UCP.
It allows to treat non-linear perturbations in a similar way as linear ones, at least in the weak coupling regime.
The method of proof requires some restrictions on our model, namely a monotonicity assumption, as formulated in \eqref{eq:partial_u}, and the weak coupling regime.
\par
Let us summarize previous results on random breather Schr\"odinger operators: They have been introduced in the mathematical literature in \cite{CombesHM-96}, a non-trivial Wegner estimates was proven in \cite{CombesHN-01}, Lifschitz tail estimates were established in \cite{KirschV-10} and \cite{Veselic-07}, and a flexible  Wegner estimate was given in \cite{NakicTTV-15}, which covers also \emph{standard} random breather
potentials $ V_\omega(x) := \sum_{j \in \ZZ^2} u_{\omega_j}(x-j)$ where the single site potential $u$ is the characteristic function of a ball or a cube.
This model is excluded in the present paper, since it does not fullfill hypothesis \eqref{eq:partial_u}.
We should also mention that the analysis of random breather models shares many features and challenges with the ones of Delone-alloy potentials, studied e.g. in \cite{RojasMolinaV-13,RojasMolina-Thesis,Klein-13}.
However it is fair to say that random breather models are more difficult to analyze, precisely due to the mentioned non-linearity, cf.~the discussion in the last paragraph in \cite{NakicTTV-15}.
In the next section we formulate our results, then we infer the UCP we are using, and in the last section we prove the main theorem on the Wegner estimate and the mentioned abstract Theorem \ref{prop:prop_1}.
\section{Notation and main result}

Througout this paper we define $\Lambda_L(x) := x + (-L/2, L/2)^2 \subset \RR^2$ as the open cube of side length $L$ and $B(x,r)$ as the open ball of radius $r$ centered at $x \in \RR^2$.
If $x = 0$, we simply write $\Lambda_L$.
%
The Landau Hamiltonian is the self-adjoint operator
$ H_B = (- i \nabla - A)^2$ with $A = (B/2) (x_2, -x_1)$ on $L^2(\RR^2)$ where $B>0$ is the magnetic field strength.
We define a scale $L_B > 0$ such that corresponding squares have integer flux by letting
\[
 K_B := 2 \lceil \sqrt{B/(4 \pi)} \rceil,
 \quad L_B = K_B \sqrt{4 \pi / B},
 \quad
 \text{and}
 \quad
 \NN_B = L_B \NN
\]
where $\lceil a \rceil$ denotes the least integer larger or equal than $a$.
%
Let $\{ \omega_j \}_{j \in \ZZ^2}$ be a process of i.i.d random variables on a probability space $(\Omega, \PP) = (\times_{\ZZ^d} (0, 1/2), \otimes_{\ZZ^d} \mu)$, where $\mu$ is a probability measure supported on $[\omega_{-}, \omega_{+}] \subset (0, 1/2)$ with bounded density $\nu_\mu$.
%
For the single-site potential $u: \RR^d \to [0, \infty)$ we make the following assumptions:
\begin{enumerate}[(i)]
 \item $u$ is measurable, bounded and compactly supported.
 \item For every $t \in [\omega_{-}, \omega_{+}]$, the map $x \mapsto (\partial / \partial t) u(x / t)$ exists  for almost every $x \in \RR^d$.
 \item There is $C_u >0$ such that for every $t \in [\omega_-, \omega_+]$ we find $x_0 = x_0(t) \in \Lambda_{1}$ with
\begin{equation}
\label{eq:partial_u}
  \frac{\partial}{\partial t} u \left( \frac{x}{t} \right) \geq C_u \chi_{B(x_0(t),r)}(x)
  \
  \text{for almost every}
  \
  x \in \RR^d.
\end{equation}
\end{enumerate}
In Remark~\ref{rem:u} we comment on these assumptions and provide some explicit examples.
Now, we define the random breather potential as
\[
 V_\omega(x) := \sum_{j \in \ZZ^2} u_{\omega_j}(x-j)
 \quad
 \text{where}
 \quad
 u_t(x) := u \left( \frac{x}{t} \right).
\]
The Landau-breather Hamiltonian is the family of operators
$
 \left\{ H_B^\omega = H_B + \lambda V_\omega \right\}_{ \omega \in \Omega},
$
where $\lambda > 0$ is the disorder parameter.
Let $H_{B,L}$ and $H_{B,L}^\omega$ be restrictions of the operators $H_B$ and $H_B^\omega$ to $L^2(\Lambda_L)$ with periodic boundary conditions, respectively.

\begin{theorem}[Wegner estimate]
\label{thm:Wegner}
Assume that $u \in L^\infty(\RR^2)$ satsifies the hypotheses (i)--(iii) above and let $B > 0$, $E_0 \in \RR$, $\theta \in (0,1)$.
Then there is $\lambda_0 > 0$ such that for all $0 < \lambda < \lambda_0$ we find $C = C(B,E_0,\theta) > 0$ and $L_0 \in \NN_B$ such that for all intervals $I \subset (- \infty, E_0]$ with $\lvert I \rvert \leq B/2$ and all $L \in \NN_B$ with $L \geq L_0$ we have
\[
 \EE \left[ \operatorname{Tr} \left( \chi_{I}(H_{B,L}^\omega) \right) \right]
 \leq
 C \cdot \lvert I \rvert^\theta \cdot L^2 .
\]
The critical disorder parameter $\lambda_0$ only depends on $B$, $E_0$ and  $u$ and is explicitely given in inequality \eqref{eq:lambda}.
\end{theorem}
As a corollary, we obtain H\"older continuity of the integrated density of states (IDS) in the low disorder regime.
For an ergodic random, self-adjoint operator $H$ on $L^2(\RR^2)$ and its restrictions $H_L$ to boxes $\Lambda_L$ with self-adjoint boundary conditions,
the IDS is defined as the almost sure limit
\[
 N(E,H) = \lim_{L \to \infty} \frac{ \operatorname{Tr} \left( \chi_{(- \infty, E]}(H_L) \right)}{L^2},
\]
see \cite{PasturF-92} and \cite{Veselic-08} for a broader discussion. In fact, the almost sure limit is independent of the choice of $\omega$;
a consequence of ergodicity.
Although the Landau Hamiltonian $H_B$ is not random, its IDS $N(\cdot, H_B)$  is well defined.
In fact, it is an explicitly calculable step function with jumps at the Landau levels $B (2 n - 1)$, $n \in \NN$, cf. \cite{HupferLMW-01b,RaikovW-02b}.
For random families of operators, the IDS needs not to exist a priori, but if the family is ergodic,
as it is in our case for the Landau-breather Hamiltonian, then it will exist almost surely and will be independent of $\omega$, cf. \cite{Veselic-08} and the references therein.
We denote by $N(\cdot, H_\omega)$ the IDS of the ensemble $\{H_B^\omega\}$.

\begin{corollary}
\label{cor:IDS}
 Fix $B > 0$, $E_0 \in \RR$ and $\theta \in (0,1)$.
 Then, for all disorder parameters $0 < \lambda \leq \lambda_0$ we have
 \[
  N(E, H_B^\omega) - N(E- \epsilon, H_B^\omega) \leq C \cdot \lvert \epsilon \rvert^\theta
  \quad
  \text{for all}\
  \epsilon > 0,\
  E \leq E_0.
 \]
\end{corollary}
Thus, the IDS of $H_B^\omega$ is locally H\"older continuous in $(- \infty, E_0]$ with respect to any H\"older exponent.

\begin{remark}
\label{rem:u}
Note that in condition~\eqref{eq:partial_u} the radius $r$ and the constant $C_u$ need to be $t$-independent, while we can allow $x_0$ to vary with $t$.
Condition \eqref{eq:partial_u} translates into
 \[
  -x/t^2 \cdot (\nabla u) (x/t)
  \geq
  C_u \chi_{B(x_0,r)}(x)  \quad \text{ for all } x \in \RR^d,\ t \in [ \omega_{-}, \omega_{+}]
 \]
 or equivalently
 \[
  -y  \cdot \nabla u(y)
  \geq
  C_u t\chi_{B(x_0,r)}(ty)
  =
  C_u t\chi_{\frac{1}{t}B(x_0,r)}(y)
  \quad \text{ for all } y,t
 \]
 and can be compared to the conditon on the breather potential
 \[
  u \in C^1(\RR^d\setminus \{0\}), \quad -x \cdot \nabla u \geq  \epsilon_0 u  \ \text{for all}\ x \in \RR^d \backslash \{ 0 \}
 \]
with fixed $\epsilon_0>0$  in \cite{CombesHN-01} which implies a singularity at the origin which we do not have.
\par
Let us give two examples of single-site potentials $u$ satisfying our assumptions.
 \begin{itemize}
 \item The smooth function
  \[
   u(x) = \exp \left( - \frac{1}{1 - \lvert x \rvert^2} \right) \chi_{\lvert x \rvert < 1},
  \]
  since, using $\omega_{-} \leq t \leq  1/2$
  \begin{align*}
  -x/t^2 \cdot (\nabla u) (x/t)
   &=
%
   x/t^2 \cdot \exp \left( - \frac{1}{1 - \lvert x/t \rvert^2} \right)
   \frac{2 (x/t^2) }{(1 - \lvert x/t \rvert^2)^2}
   \cdot
   \chi_{\lvert x \rvert < t}\\
   &\geq
   2 \euler^{-4/3}
   \lvert x \rvert^2
   \cdot
   \chi_{\lvert x \rvert < t/2}
   \geq
   \frac{\lvert x \rvert^2}{2} \chi_{\omega_{-}/4 \leq \lvert x \rvert < \omega_{-}/2}
   \geq
   \frac{\omega_{-}^2}{32} \chi_{B(x_0, \omega_{-}/8)}
  \end{align*}
  for every point $x_0$ with $\lvert x_0 \rvert = 3 \omega_{-}/8$.
  \item The hat potential $u(x) = \chi_{\lvert x \rvert < 1} (1 - \lvert x \rvert)$, since
  \[
   -x/t^2 \cdot (\nabla u) (x/t)
   =
   \frac{ \lvert x \rvert}{t^3} \chi_{\lvert x \rvert < t}
   \geq
   \frac{\chi_{t/2 \leq \lvert x \rvert < t}}{2 t^2}
   \geq
   \frac{1}{2 \omega_{+}^2} \chi_{B(x_0, \omega_{-}/4)}
   \geq
   2\chi_{B(x_0, \omega_{-}/4)}
  \]
  for every point $x_0$ with $\lvert x_0 \rvert = 3 t/4$.
 \end{itemize}
\end{remark}

\section{Unique continuation principle}
The spectrum of the Landau Hamiltonian on the torus $H_{B,L}$ consists of an increasing sequence of isolated eigenvalues of finite multiplicity at the Landau Levels $B(2 n - 1)$, $n = 1,2,...$, see for instance \cite{GerminetKS-07}, Section~5.
We denote the spectral projector onto the $n$-th Landau level by $\Pi_{n,L}$.
Furthermore, for $L \in \NN_B$ and $x \in \RR^2 / L \ZZ^2$, we write $\hat \chi_{x,L}$ for the characteristic function of the cube with side length $L$, centered at $x$ on the torus $\RR^2 / L \ZZ^2$.
The following Lemma is \cite{GerminetKS-07}, Lemma 5.3, which is an adaptation of \cite{CombesHKR-04}, Lemma 2, which itself is based on bounds developed in \cite{RaikovW-02a,RaikovW-02b}.

\begin{lemma}
\label{lem:GKS}
Fix $B > 0$, $n \in \NN$, $R > r > 0$ and $\eta > 0$. If $\kappa > 1$ and $L \in \NN_B$ are such that $L > 2 (L_B + \kappa R)$ then for all $x \in \Lambda_L$ we have
\[
 \Pi_{n,L} \hat \chi_{x, r} \Pi_{n,L} \geq C_0 \Pi_{n,L} \left( \hat \chi_{x,R} - \eta \hat \chi_{x, \kappa R} \right) \Pi_{n,L} + \Pi_{n,L} \tilde{ \mathcal{E}}_x \Pi_{n,L}
\]
where $C_0 = C_0(n,B,r, R, \eta) > 0$ is a constant and the symmetric error operator
\[
  \tilde{ \mathcal{E}}_x = \tilde{ \mathcal{E}}_x(n,L,B,r,R,\eta)
\]
satisfies
\[
 \lVert \tilde{ \mathcal{E}}_{x} \rVert \leq C_{n,B,r,R,\eta} \euler^{- m_{n,B} L}
\]
for some constants $C_{n,B,r,R,\eta} > 0$ and $m_{n,B} > 0$ which in particular do not depend on $x$.
\end{lemma}
\begin{definition}
 Let $r \in (0, 1/2)$. A sequence $\{ x_j \}_{j \in \ZZ^2}$ is called $r$-equidistributed if for every $j \in \ZZ^2$ we have $B(x_j,r) \subset \Lambda_1(j)$.
 Given a $r$-equidistributed sequence and $L > 0$ we define $W_r(L) := \sum_{j \in \ZZ^2: B(x_j,r) \subset \Lambda_L} \chi_{B(x_j,r)}$.
\end{definition}

We use the previous lemma for the following observation:
\begin{lemma}
\label{lem:ucp}
 Fix $B > 0$, $n \in \NN$.
There are $C_1 = C_1(n,B,r) = C_0(n,B,r,4,1/162)/4$, $L_0 = L_0(n,B,r) > 0$ such that for all $L \in \NN_B$ with $L \geq L_0$, all $ r \in (0, 1/2)$ and all $r$-equidistributed sequences we have
 \[
  \Pi_{n,L} \chi_{W_r(L)} \Pi_{n,L} \geq C_1 \Pi_{n,L} .
 \]
\end{lemma}
\begin{remark}
An essentially equivalent bound can be found in \cite[Lemma 3.2.3]{RojasMolina-Thesis}, where it is used for a Wegner estimate for the Delone-alloy-type model using the method of \cite{CombesHK-07}.
\end{remark}

\begin{proof}
We choose a large $L \in \NN_B$ to be determined later and will apply Lemma \ref{lem:GKS} with $r$, $R = 4$, $\kappa = 2$ and $\eta = 1/162$.
 Recall that $r < 1/ 2 < R$.
 We estimate
 \begin{align*}
  & \Pi_{n,L} \chi_{W_r(L)} \Pi_{n,L}
   \geq \sum_{j \in \ZZ^2 : B(x_j,r) \subset \Lambda_L } \Pi_{n,L} \hat \chi_{x_j,r} \Pi_{n,L}\\
  &\geq C_0 \sum_{j \in \ZZ^2 : B(x_j,r) \subset \Lambda_L } \left( \Pi_{n,L} ( \hat \chi_{x_j,4 } - \eta \hat \chi_{x_j, 8 } ) \Pi_{n,L} + \Pi_{n,L} \tilde{\mathcal{E}}_{x_j} \Pi_{n,L} \right) .
 \end{align*}
 Since $x_j \in \Lambda_{1}(j)$ for all $j \in \ZZ^2$, we have for $L \geq 3$ that for every $x \in \Lambda_L$ there is $j \in \ZZ^2$ with $B(x_j,r) \subset \Lambda_L$ such that $x \in \Lambda_{4}(x_j)$. Therefore,
 \[
  \bigcup_{j \in \ZZ^2 : B(x_j,r) \subset \Lambda_L } \hat \Lambda_{4}(x_j) \supset \RR^2 / L \ZZ^2,\ \text{the}\ L-\text{torus}.
 \]
Furthermore, given $x \in \Lambda_L$, there are at most $81$ elementary cells $\Lambda_{1}(j)$ of $\ZZ^2$ in which $\hat \chi_{x_j, 8}$ can be non-zero.
Hence, we can bound the sum from below by
 \begin{equation}
 \label{eq:lemma_2.2}
  C_0 \left( \Pi_{n,L} - \eta 81 \Pi_{n,L} \right) + \sum_{j \in \ZZ^2 : B(x_j,r) \subset \Lambda_L } \Pi_{n,L} \tilde{ \mathcal{E}}_{x_j} \Pi_{n,L} = \frac{C_0}{2} \Pi_{n,L} + \Pi_{n,L} \mathcal{E}_L \Pi_{n,L}
 \end{equation}
 with a symmetric error operator $\mathcal{E}_L$ satisfying
 \[
  \lVert \mathcal{E}_L \rVert
   \leq C_{n,B} \left( \frac{L}{L_B} \right)^2 \euler^{- m_{n,B} L}.
 \]
This implies that there is $ \tilde L_0 > 0$ such that for all $L \in \NN_B$ with $L \geq \tilde L_0$ we have $\lVert \mathcal{E}_L \rVert \leq C_0/4$ whence in particular $\mathcal{E}_L \geq - C_0/4 \cdot \operatorname{Id}$ in the operator sense.
Since we used $L \geq 3 $ and since inequality~\eqref{eq:lemma_2.2} requires $L \geq 2 ( L_B + 8)$, we need $L \geq L_0 := \max \{ \tilde L_0, 2 ( L_B + 8), 3 \}$ to deduce the estimate
 \[
  \Pi_{n,L} \chi_{W_r(L)} \Pi_{n,L} \geq \frac{C_0}{4} \Pi_{n,L} .
  \qedhere
 \]
\end{proof}

\section{Proof of Theorem~\ref{thm:Wegner}}

We start the proof with the following abstract theorem:

\begin{theorem}
\label{prop:prop_1}
 Let $H$ be a lower semibounded self-adjoint operator with purely discrete spectrum,
 $V$ a bounded symmetric operator, and  $I \subset  J \subset  \RR$  two intervals.
 We assume that there are $C_2 > 0$ and a positive, symmetric operator $W$ such that
 \begin{equation}
 \label{eq:ucp_with_W}
  \chi_{J}(H) W \chi_{J}(H) \geq C_2 \chi_{J}(H).
 \end{equation}
 Then, for $\lVert V \rVert < \operatorname{dist}(I, J^c) \sqrt{C_2/(C_2 + 1 + \lVert W \rVert)}$ there is $C_3$ depending only on $C_2$, $\operatorname{dist}(I, J^c)$ and on $\lVert V \rVert$ such that
 \begin{equation}
 \label{eq:result_prop_1}
  \operatorname{Tr} \left[ \chi_{I}(H+V) \right]
  \leq
  C_3 \operatorname{Tr}(\chi_I(H+V) (W + W^2))
 \end{equation}
More precisely, we have
\[
  C_3 = \frac{\operatorname{dist}(I, J^c)^2}{C_2 \operatorname{dist}(I, J^c)^2 - \lVert V \rVert^2(C_2 + 1 + \lVert W \rVert)} .
\]
\end{theorem}
Inequality \eqref{eq:ucp_with_W} tells us that $W$ is strictly positive on a spectral subspace of $H$, or
equivalently that vectors in the spectral subspace of $H$ cannot be \emph{completely localized with respect to $W$}.

\begin{proof}
 We decompose
 \begin{equation}
  \label{eq:decomposition}
  \operatorname{Tr}(\chi_{I}(H+V)) =   \operatorname{Tr}(\chi_I(H+V)\chi_{J}(H)) + \operatorname{Tr}(\chi_I(H+V) \chi_{J^c}(H)).
 \end{equation}

 We estimate the term in \eqref{eq:decomposition} containing $\chi_{J^c}(H)$ by expanding the trace in eigenfunctions $\phi_j$ in the range of $\chi_I(H+V)$.
 From the eigenvalue equation $(H + V - E_j) \phi_j = 0$ we deduce
 \[
  - (H - E_j)^{-1} \chi_{J^c}(H) V \chi_j = \chi_{J^c}(H) \phi_j.
 \]
 This yields
 \begin{align}
 \label{eq:I^c}
  \operatorname{Tr} \left( \chi_I(H+V) \chi_{J^c}(H) \right)
 & = \sum_j \left\langle \phi_j, \chi_{J^c}(H) \phi_j \right\rangle = \\
 \nonumber
  \sum_j \left\langle \phi_j, \left( V \frac{\chi_{J^c}(H)}{(H - E_j)^2} V \right) \phi_j \right\rangle
  &\leq
    \frac{\lVert V \rVert^2}{\operatorname{dist}(I, J^c)^2} \operatorname{Tr} \left( \chi_I(H+V) \right) .
 \end{align}
 Now we turn to the first summand on the right hand side of \eqref{eq:decomposition}.
 Using the assumption \eqref{eq:ucp_with_W}, we have
 \begin{align*}
  \operatorname{Tr} ( \chi_I(H+V) \chi_{J}(H))
  & \leq \frac{1}{C_2} \operatorname{Tr} (\chi_I(H+V) \chi_{J}(H) W \chi_{J}(H)) \\
  & =\frac{1}{C_2} [ \operatorname{Tr} (\chi_I(H + V) W) + \operatorname{Tr} (\chi_I(H+V) \chi_{J^c}(H) W \chi_{J^c}(H) ) \\
  & \quad  - 2 \operatorname{Re}(\operatorname{Tr} \left[ \chi_I(H+V) \chi_{J^c}(H) W) \right]\\
  & \leq \frac{1}{C_2} [ \operatorname{Tr} (\chi_I(H + V) W) + \lVert W \rVert \operatorname{Tr} ( \chi_I(H+V) \chi_{J^c}(H))\\
  & \quad +  \operatorname{Tr}(\chi_I(H+V) \chi_{J^c}(H)) + \operatorname{Tr}(\chi_I(H+V) W^2)].
 \end{align*}
In the last step, we used $- \operatorname{Re}(x) \leq \lvert x \rvert$, cyclicity of the trace, the Hoelder inequality for traces, and the fact that $2 ab \leq a^2 + b^2$ to estimate
\begin{align*}
- 2 \operatorname{Re}(\operatorname{Tr} ( \chi_I(H+V) \chi_{J^c}(H) W)
& \leq
 2 \lvert \operatorname{Tr} \left( \chi_I(H+V) \chi_{J^c}(H) W \chi_I(H+V) \right) \rvert\\
& \leq
 \operatorname{Tr} ( \chi_I(H+V) \chi_{J^c} )
 +
 \operatorname{Tr} ( \chi_I(H+V) W^2 ) .
\end{align*}
This simplifies to
\begin{equation}
\label{eq:2}
 \operatorname{Tr} ( \chi_I(H+V) \chi_{J}(H) )
 \leq
 \frac{1 + \lVert W \rVert}{C_2} \operatorname{Tr} (\chi_I(H+V) \chi_{J^c}(H))
 +
 \frac{1}{C_2} \operatorname{Tr} ( \chi_I(H+V) (W + W^2)).
\end{equation}
Combining \eqref{eq:decomposition} with \eqref{eq:I^c} and \eqref{eq:2} we find
\begin{align*}
 \operatorname{Tr} ( \chi_I(H+V) )
 &\leq
 \left( 1 + \frac{1 + \lVert W \rVert}{C_2} \right) \frac{\lVert V \rVert^2}{\operatorname{dist}(I, J^c)^2} \operatorname{Tr} ( \chi_I(H+V) )\\
 & \qquad
 + \frac{1}{C_2} \operatorname{Tr} ( \chi_I(H+V) (W + W^2))
\end{align*}
that is
\[
 \operatorname{Tr} (\chi_I(H+V) ) \leq \frac{\operatorname{dist}(I, J^c)^2}{C_2 \operatorname{dist}(I, J^c)^2 - \lVert V \rVert^2(C_2  + 1 + \lVert W \rVert)} \operatorname{Tr}(\chi_I(H+V) ( W + W^2)).
\]

\end{proof}

We are now ready for the proof of Theorem~\ref{thm:Wegner}.

\begin{proof}[Proof of Theorem~\ref{thm:Wegner}]
Given $B > 0$ and $E_0 \in \RR$, there are finitely many Landau Levels below $E_0 + B/4$.
Let $L \in \NN_B$, $L \geq L_0$ as in Lemma~\ref{lem:ucp}.
Take
\[
  \tilde C := \min \{ C_1(n,B,r)\ \text{from Lemma~\ref{lem:ucp}}: n \in \NN\ \text{with}\ B (2 n - 1) \leq E_0 + B/4 \}.
\]
Let $I = [I_-, I_+] \subset (- \infty, E_0]$ with $I_+ - I_- \leq B/2$.
For every constellation $\{ \omega_j \}_{j \in \ZZ_B^2}$, we apply Proposition~\ref{prop:prop_1} with $V = \lambda V_\omega$ and $J = [I_- - B/4, I_+ + B/4]$ and $W = \cup_{j \in \ZZ^2: B(x_j + j,r) \subset \Lambda_L} \chi_{B(x_0(\omega_j) + j),r}$, where the $x_0(\omega_j)$ are the points from \eqref{eq:partial_u}.
Note that $J$ contains at most one Landau Level.
Hence, \eqref{eq:ucp_with_W} holds by Lemma~\ref{lem:ucp} with $C_2 = \tilde C$.
We find
\[
 \operatorname{Tr} ( \chi_I(H+V)) \leq  C_3 \operatorname{Tr}(\chi_I(H_{B,L}^\omega)(W)),
\]
where
\[
 C_3 = \frac{(B/4)^2}{\tilde C (B/4)^2 - \lVert V_\omega \rVert_\infty^2 \lambda^2 (\tilde C + 2)}.
\]
We have $\lVert V_\omega \rVert_\infty \leq V_\infty :=  \lceil \max \operatorname{supp} u \rceil^2 \lVert u \rVert_\infty$.
For
\begin{equation}
\label{eq:lambda}
 \lambda \leq \lambda_0 := B \sqrt{ \tilde C / (32 V_\infty^2 (\tilde C + 2))} ,
\end{equation}
it holds that $C_3 \leq 2/\tilde C$.
We have estimated so far
\[
 \operatorname{Tr} ( \chi_I(H+V)) \leq \frac{2}{\tilde C} \operatorname{Tr}(\chi_I(H_{B,L}^\omega)(W)) .
\]
Now we take a monotone decreasing function $f \in C^1(\RR)$ with $f \equiv 1$ on $(- \infty, I_- - B/4]$ and $f \equiv 0$ on $[I_+ + B/4, \infty)$ such that $- C_4 \lvert I \rvert f'(x) \geq   \chi_I(x)$ for some $C_4 > 0.$
Then
\begin{align*}
 \operatorname{Tr} (\chi_I(H_{B,L}^\omega) W )
 &\leq
 - C_4 \lvert I \rvert \operatorname{Tr} (f'(H_{B,L}^\omega) W) \\
 & \leq
 - C_u^{-1} C_4 \lvert I \rvert \sum_{j \in \ZZ^2 \cap \Lambda_L} \mathrm{Tr} ( f'(H_{B,L}^\omega) \frac{\partial}{\partial \omega_j} u_{\omega_j}(x -j))\\
 &=
  - C_u^{-1} C_4 \lvert I \rvert \sum_{j \in \ZZ^2 \cap \Lambda_L} \mathrm{Tr} (\frac{\partial}{\partial \omega_j} f(H_L^\omega)).
\end{align*}
We take the expectation and obtain
\[
 \EE \left[ \operatorname{Tr} (\chi_I(H_{B,L}^\omega) W )  \right]
 \leq
 - C_u^{-1} C_4 \lvert I \rvert \sum_{j \in \ZZ^2 \cap \Lambda_L}
  \EE \left[ \frac{\partial}{\partial \omega_j} \operatorname{Tr} ( f(H_{B,L}^\omega )) \right].
\]
We evaluate the expectation in every summand with respect to the random variable $\omega_j$
\begin{align*}
 0
 &\leq
 - \EE \left[ \frac{\partial}{\partial \omega_j} \operatorname{Tr} ( f(H_{B,L}^\omega )) \right]
 =
 - \EE \left[ \int_{\omega_-}^{\omega_+}  \frac{\partial}{\partial \omega_j} \operatorname{Tr} ( f(H_{B,L}^\omega )) \mathrm{d} \omega_j \right]\\
 &\leq
 \lVert \nu_\mu \rVert_\infty \EE
  \left[ \lvert \operatorname{Tr}
    \left( f(H_{B,L}^\omega \mid_{\omega_j = \omega_+}) - f(H_{B,L}^\omega \mid_{\omega_j = \omega_-})
    \right) \rvert
  \right]
\end{align*}
Analogously to \cite{CombesHK-03}, Appendix $A$ we find for every $\theta \in (0,1)$ a constant $C_\theta > 0$ such that
\[
 \left\lvert \operatorname{Tr}
    \left( f(H_{B,L}^\omega \mid_{\omega_j = \omega_+}) - f(H_{B,L}^\omega \mid_{\omega_j = \omega_-})
    \right) \right\rvert \leq C_\theta \lvert I \rvert^{\theta - 1}.
\]
All together we found
\[
 \operatorname{Tr} \left( \chi_I(H_{B,L}^\omega) \right)
 \leq
  \frac{4}{C_u \tilde C} C_4 \lvert I \rvert C_\theta \lvert I \rvert^{\theta - 1} \# \{ \Lambda_L \cap \ZZ^2 \} = C \lvert I \rvert^{\theta} L^2. \qedhere
\]
\end{proof}

\begin{remark}
Let us briefly discuss potential improvements of our results, first concerning the assumption on the small coupling constant.

 One candidate for replacing the smallness condition on $\lVert V \rVert$ in Theorem~\ref{prop:prop_1} which ensures positivity of $C_2 \operatorname{dist}(I, J^c)^2 - \lVert V \rVert^2(C_2 + 1 + \lVert W \rVert)$ would be a largeness condition on $\operatorname{dist}(I, J^c)$.
 In the application (i.e. in the proof of Theorem~\ref{thm:Wegner}) an upper bound on $\operatorname{dist}(I, J^c)$ arises from the fact that $J$ should contain at most one Landau Level.
 Therefore, it would be desirable to improve Lemma~\ref{lem:ucp} to something like
 \[
  \left( \sum_{k = 1}^n \Pi_{n,L} \right) W_r(L) \left( \sum_{k = 1}^n \Pi_{n,L} \right)
  \geq
  C \left( \sum_{k = 1}^n \Pi_{n,L} \right),
 \]
 where we know how $C = C(n,B,r)$ behaves asymptotically for large $n$.
 In particular we would need to have
 \[
 C(n,B,r) \gg n^{-2}.
 \]

Another critique of our result is that it gives only Hoelder continuity of the integrated density of states, not Lipschitz continuity.
 There is a Wegner estimate for the alloy-type model in \cite{CombesHK-07} where the exponent $\theta \in (0,1)$ has been replaced by the optimal $1$, but the proof given there heavily relies on linearity of the random potential.
\end{remark}

\subsubsection*{Acknowledgements}
Part of this work was done while the authors were visiting the Hausdorff Research Institute for Mathematics
during the Trimester Program \emph{Mathematics of Signal Processing}.
Financial support by the DFG through grant \emph{Unique continuation principles and equidistribution properties of eigenfunctions}.
We thank M.~Egidi for reading an earlier version of the manuscript.


\end{document}